\documentclass[a4paper,USenglish,cleveref,autoref,thm-restate]{lipics-v2021}

\title{On the Doubling Dimension and the Perimeter of Geodesically Convex Sets in Fat Polygons} 

\author{Mark de Berg}{Department of Mathematics and Computer Science, TU Eindhoven, the Netherlands}{M.T.d.Berg@tue.nl}{https://orcid.org/0000-0001-5770-3784}{MdB is supported by the Dutch Research Council (NWO) through Gravitation-grant NETWORKS-024.002.003.}
\author{Prosenjit Bose}{School of Computer Science, Carleton University, Canada}{jit@scs.carleton.ca}{https://orcid.org/0000-0002-8906-0573}{PB is supported in part by the  Natural Sciences and Engineering Research Council of Canada (NSERC).}
\author{Leonidas Theocharous}{School of Electrical Engineering and Computer Science, University of Ottawa, Canada}{leotheocharous3@gmail.com}{https://orcid.org/0000-0002-1707-6787}{LT is supported in part by the  Natural Sciences and Engineering Research Council of Canada (NSERC).}
\authorrunning{M.~de Berg and P.~Bose and L.~Theocharous} 
\Copyright{Mark de Berg and Prosenjit Bose  and Leonidas Theocharous} 
\nolinenumbers

\usepackage{amssymb,amsmath,mathtools,microtype,xspace}
\newcommand{\Reals}{\mathbb{R}}

\newcommand{\bd}{\partial}

\newcommand{\myint}{\mathrm{int}}  

\newcommand{\diam}{\mathrm{diam}}
\newcommand{\radius}{\mathrm{radius}}

\newcommand{\area}{\mathrm{area}}
\newcommand{\rch}{\mbox{\sc rch}}
\newcommand{\vd}{\mbox{\sc vd}}

\renewcommand{\geq}{\geqslant}
\renewcommand{\leq}{\leqslant}
\renewcommand{\ge}{\geqslant}
\renewcommand{\le}{\leqslant}

\newcommand{\etal}{et al.}

\newcommand{\eps}{\varepsilon}
\newcommand{\gdiam}{\mathrm{diam}_\mathrm{g}}
\newcommand{\per}{\mathrm{per}}
\newcommand{\fn}{\mbox{\sc fn}}
\newcommand{\gdist}{\mathrm{dist}_{\mathrm{g}}}     

\newcommand{\C}{\ensuremath{\mathcal{C}}}
\newcommand{\D}{\ensuremath{\mathcal{D}}}

\newcommand{\X}{\ensuremath{\mathcal{X}}}

\newtheorem{fact}[claim]{Fact}

\newif\ifcomments
 \commentsfalse        

\ifcomments
  \newcommand{\jit}[1]{\textcolor{blue}{Jit: #1}}
  \newcommand{\lt}[1]{\textcolor{teal}{Leo: #1}}
  \newcommand{\mdb}[1]{\textcolor{red}{Mark: #1}}
\else
  \newcommand{\jit}[1]{}
  \newcommand{\lt}[1]{}
  \newcommand{\mdb}[1]{}
\fi

\begin{document}
\maketitle

\begin{abstract}
Many algorithmic problems can be solved (almost) as efficiently in metric spaces of bounded doubling
dimension as in Euclidean space. Unfortunately, the metric space defined by points in a simple polygon equipped with
the geodesic distance does not necessarily have bounded doubling dimension. We therefore
study the doubling dimension of fat polygons, for two well-known fatness definitions.
We prove that locally-fat simple polygons \emph{do not} always have bounded doubling dimension,
while any $(\alpha,\beta)$-covered polygon \emph{does} have bounded doubling dimension
(even if it has holes).
We also study the perimeter of geodesically convex sets in $(\alpha,\beta)$-covered polygons (possibly with holes),
and show that this perimeter is at most a constant times the Euclidean diameter of the set.

Using these two results, we obtain new results for several problems on $(\alpha,\beta)$-covered polygons,
including an algorithm that computes the closest pair of a set of $m$ points in an
$(\alpha,\beta)$-covered polygon with $n$ vertices that runs in $O(n + m\log n)$ expected time. 
\end{abstract}

\section{Introduction}

\subparagraph{Motivation and background.}
In computational geometry, the standard way to model planar objects or regions is by
polygons. Because polygons can have very complicated shapes and/or
interact with each other in complicated ways, developing efficient solutions
for algorithmic problems involving polygons is not always easy. In practice, however, the objects 
or regions under consideration are usually fairly well-behaved. The \emph{realistic input models} 
paradigm~\cite{DBLP:journals/algorithmica/BergSVK02} suggests to make this precise by 
defining properties of the input that are expected to be satisfied in practice and that 
allow for easier and/or more efficient algorithmic solutions. 

One such property that has been studied extensively is \emph{fatness}, where the underlying assumption is 
that real-world objects or regions are not extremely thin;
see below for more precise definitions. For example, there have been results on motion-planning 
in environments where the obstacles are fat~\cite{DBLP:journals/comgeo/StappenHO93,DBLP:journals/dcg/StappenOBV98} 
on various graphics-related problems involving fat 
objects~\cite{DBLP:journals/comgeo/AgarwalKS95,DBLP:journals/comgeo/KatzOS92,DBLP:journals/comgeo/BergG10a,DBLP:journals/siamcomp/BergG08}, 
on flow problems on terrains consisting of fat triangles~\cite{DBLP:journals/comgeo/BergCHLT10}, 
on various data-structure problems for fat 
objects~\cite{DBLP:journals/comgeo/EfratKNS00,DBLP:journals/comgeo/Katz97,DBLP:journals/jal/OvermarsS96}, and more.
Many of these results use one of the following two fundamental results: (i)~the union complexity of $n$ 
constant-complexity fat polygons in the plane is near-linear; see~\cite{DBLP:journals/siamcomp/AronovBES14} and
the references therein; and
(ii)~any set of $n$ constant-complexity disjoint fat objects in the plane admits a  hierarchical
decomposition of the plane into constant-complexity cells that each intersect $O(1)$ objects~\cite{DBLP:journals/algorithmica/Berg00}.
Note that both results concern constant-complexity objects and that they
essentially say something about the part of the space outside the objects. We are interested in what 
happens \emph{inside} a fat polygon. 
In particular, we are interested in properties of the metric space
$(P,\gdist)$ where $P$ is a fat polygon (possibly with holes) and $\gdist$ is the geodesic distance.
In other words, $\gdist(p,q)$ is the length of a shortest path between two points~$p$ and $q$ in $P$.
There has been work on guarding and triangulating fat polygons~\cite{DBLP:journals/comgeo/AloupisBDGLS14}, 
but the special properties that the metric space $(P,\gdist)$ may or may not have for fat polygons
have not been studied much, to the best of our knowledge. This is surprising as 
problems involving shortest paths in polygons have received widespread attention,
and fatness has been studied and used extensively as well.
The only paper that is closely related to our work is by Bose, Cheong, and Dujmovic~\cite{bose11},
who study the perimeter of fat objects, as discussed below.

\subparagraph{Fatness variants.}
The study of fat polygons was initiated by Matou\v{s}ek~\etal~\cite{DBLP:journals/siamcomp/MatousekPSSW94}, 
who studied the union complexity of \emph{$\alpha$-fat triangles}, which are triangles whose angles 
are at least~$\alpha$, for some fixed constant~$\alpha>0$. Since then, the concept of fatness
has been generalized in different ways, both to convex and to non-convex objects.
For convex objects, all definitions are equivalent
in the sense that a convex object that is fat under one definition is also fat under the other
definitions~\cite{DBLP:books/daglib/0084325,thesis-vleugels}.
Our interest lies in non-convex objects, since inside a convex object the
geodesic distance is identical to the Euclidean distance. For non-convex objects,
different fatness definitions have been proposed as well, but these are no longer equivalent.
We discuss the most popular ones, focusing on definitions that apply to polygons.

Overmars, Halperin, and Van der Stappen~\cite{DBLP:journals/comgeo/StappenHO93} define a polygon~$P$ to be \emph{$\gamma$-fat} if, 
for any Euclidean disk~$D$ centered inside $P$ and not fully containing~$P$, 
we have $\area(P\cap D) \geq \gamma \cdot \area(D)$, for some fixed constant~$\gamma>0$.
This definition still allows the polygon to have very thin parts, as shown in \autoref{fig:fatness}(i).
\begin{figure}
\begin{center}
\includegraphics{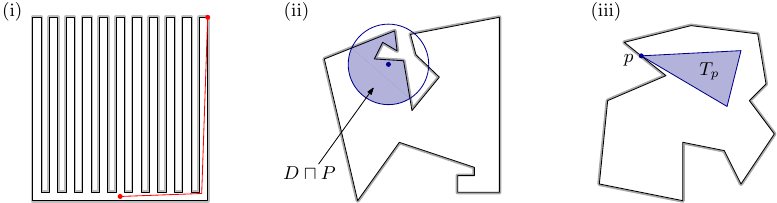}
\end{center}
\caption{(i)~A $\gamma$-fat polygon that is not locally $\gamma$-fat. 
(ii)~A locally $\gamma$-fat polygon. (iii) An $(\alpha,\beta)$-covered polygon.}
\label{fig:fatness}
\end{figure}
As a result, a collection of $n$ $\gamma$-fat, constant-complexity polygons can have
quadratic union complexity. Hence, De~Berg~\cite{DBLP:journals/dcg/Berg08} introduced 
\emph{locally $\gamma$-fat polygons}; here we require that for any Euclidean disk~$D$ 
centered inside~$P$ and not fully containing~$P$, we have 
$\area(P\sqcap D) \geq \gamma \cdot \area(D)$, where $D\sqcap P$ denotes the 
connected component of~$D\cap P$ that contains the center of~$D$; see \autoref{fig:fatness}(ii).

A third definition, by Efrat~\cite{DBLP:journals/siamcomp/Efrat05}, is the following: a polygon is
\emph{$(\alpha,\beta)$-covered} if for every $p$ on the boundary $\partial P$, 
there exists an $\alpha$-fat triangle $T_p\subset P$ with $p$ as a vertex
such that each side of $T_p$ has length at least $\beta\cdot \diam(P)$, where $\alpha$
and $\beta$ are fixed positive constants.
We will refer to such a triangle $T_p$ as a \emph{witness triangle} for the point~$p$;
see  \autoref{fig:fatness}(i). This is a stronger condition than local fatness: any $(\alpha,\beta)$-covered polygon is 
locally $\gamma$-fat for a constant~$\gamma$ that depends only on the constants~$\alpha$ and~$\beta$,
but the reverse is not true~\cite{DBLP:journals/dcg/Berg08}.

The three definitions above also apply to non-polygonal domains, even if they have holes.

\subparagraph{Our contribution.}
Recall that a metric space $(\X,d)$ is \emph{$c$-doubling} if every ball of radius~$r$ 
in the space~$\X$ can be covered by $c$ balls of radius $r/2$. The \emph{doubling dimension}
of the metric space is defined to be~$\log_{2}c$. If $c$ is a constant, then we say that 
the metric space has \emph{bounded doubling dimension}. Spaces of bounded doubling dimension
generalize Euclidean space of constant dimension, and many problems can be solved (almost)
as efficiently on spaces of bounded doubling dimension as they can in Euclidean space.
Unfortunately, if $P$ is an arbitrary polygon then the metric space~$(P,\gdist)$ may not have bounded
doubling dimension. The first problem we study is therefore: does the metric space~$(P,\gdist)$
necessarily have bounded doubling dimension when~$P$ is fat?

For the fatness definition of Overmars ~\etal~\cite{DBLP:journals/comgeo/StappenHO93},
the answer is easily seen to be no, as exemplified by the polygon in \autoref{fig:fatness}(i):
the ball centered at the red point in the middle of the polygon and whose radius~$r$ 
is the geodesic distance to the top-right vertex, needs $\Omega(n)$ geodesic balls of radius~$r/2$
to be covered. Our first two results, presented in \autoref{sec:dd}, concern the doubling dimension
of locally-fat polygons and $(\alpha,\beta)$-covered polygons.
\begin{itemize}
\item We show that locally-fat polygons do not necessarily have bounded doubling dimension.
      More precisely, we show that for any $n\geq 3$ there exists a simple polygon $P$ with $n$ vertices
      and a geodesic disk~$D$ in $P$ such that $\Omega(n^{1/3})$ geodesic disks of radius~$r$
      are needed to cover~$D$, where $r = \radius(D)$.
\item On the positive side, we prove that the doubling dimension of any $(\alpha,\beta)$-covered polygon 
      is upper bounded by a constant that only depends on the constants $\alpha$ and $\beta$ and 
      not on the number of vertices of the polygon. This result even holds for 
      $(\alpha,\beta)$-covered domains with curved boundaries and holes. 
\end{itemize}
The fact that $(\alpha,\beta)$-covered polygonal regions have bounded doubling dimension  
immediately implies a plethora of results that improve upon known results for non-fat
polygons; we mention several of them at the end of \autoref{sec:dd}.

After studying the doubling dimension of fat polygons, we turn our attention 
in \autoref{sec:geoconvex} to the following problem: can the perimeter~$\per(R)$ of a 
geodesically convex polygonal region~$R$ be bounded in terms of its 
geodesic diameter $\gdiam(R)$, or perhaps even in terms of its Euclidean diameter~$\diam(R)$? 
For arbitrary polygons the answer is clearly no. In the polygon of \autoref{fig:fatness}(i) the perimeter $\per(P)$ of $P$,
which is geodesically convex, is $ \Omega(n \cdot \gdiam(P))$, and hence arbitrarily larger than~$\diam(P)$. Bose, Cheong, and Dujmovi\'c~\cite{bose11} showed that the same can happen for locally-fat polygons, via the Koch snowflake.
For $(\alpha,\beta)$-covered simple polygons, on the other hand, they showed
that $\per(P)=O(\diam(P))$, where the constant of proportionality only depends on the 
constants $\alpha$ and $\beta$ and not on the number of vertices of the polygon. 
The applicability of this result is limited, because it only
bounds the perimeter of the polygon itself and not of geodesically convex sets
inside the polygon. We show that the result also holds in this much more general setting.
\begin{itemize}
\item For any geodesically convex polygonal set $R$ in an $(\alpha,\beta)$-covered polygonal domain~$P$,
      we have $\per(R) = O(\diam(R))$. 

\end{itemize}
A fairly immediate consequence of this result is the existence of an
$\eps$-coreset of size $O(1/\eps)$ for furthest-neighbor queries in an $(\alpha,\beta)$-covered 
simple polygon~$P$. This improves on the $O(1/\eps^2)$ bound recently proved by
De~Berg and Theocharous~\cite{DBLP:conf/compgeom/BergT24} for general simple polygons. 
The result can also be used to obtain
an algorithm for the closest-pair problem on a set $S$ of $m$ points in~$P$ whose expected
running time is linear in $m$, as explained in \autoref{sec:closest-pair}.

\subparagraph{Notation.}
We denote the Euclidean disk of radius~$r$ centered at a point~$p\in\Reals^2$ 
by $D(p,r)$ and we denote the Euclidean distance between two points~$p,q$
by~$|pq|$. For a polygon~$P$, which will always be clear from the context,
and two points $p,q\in P$, we use $\pi(p,q)$ to denote the shortest path between $p$ and 
$q$ inside $P$, and we denote the length of this path by $\|\pi(p,q)\|$. Thus,
$\gdist(p,q) = \|\pi(p,q)\|$. The geodesic disk of radius $r$ centered at a point~$p\in P$ 
is denoted by $D_{\mathrm{g}}(p,r)$. Recall that $\diam(P)$ denotes the Euclidean diameter of $P$
and that $\gdiam(P)$ denotes the geodesic diameter.
From now on, when we speak of the doubling dimension of $P$, we always
refer to the doubling dimension of the metric space $(P,\gdist)$.

\section{The doubling dimension of fat polygons}
\label{sec:dd}
In this section we investigate the doubling dimension of fat polygons.
We first show that locally-fat polygons do not necessarily have bounded
doubling dimension, and then we show that $(\alpha,\beta)$-covered
polygons do have bounded doubling dimension.

\subsection{Locally-fat polygons do not have bounded doubling dimension}
To prove that locally-fat polygons do not have bounded doubling dimension,
we first construct a family of locally-fat polygons $P_m$, for $m\geq 1$, such that 
$\diam(P_m) = \sqrt{2}$ and $\gdiam(P_m) \approx 2^m$.
The polygon $P_m$ will be constructed recursively inside the unit square $[0,1]^2$,
in such a way that the shortest path from the lower-left corner~$(0,0)$ 
to the lower-right corner~$(1,0)$ resembles the order-$m$ Hilbert curve.
\medskip

The initial polygon $P_1$, illustrated in \autoref{fig:locally-fat-construction}(i),
is constructed as follows. Take the unit square $[0,1]^2$ and a small value~$\eps>0$, and add 
the horizontal segment $[\eps,1-\eps]\times \tfrac{1}{2}$ and the vertical
segments $\tfrac{1}{2} \times [0,\tfrac{1}{2}]$ and $\tfrac{1}{2} \times [\tfrac{1}{2}+\eps,1]$.
Then slightly inflate these segments to give them width $\eps$, thus obtaining a non-degenerate
simple polygon. Note that the length of the shortest path from $(0,0)$ to $(1,0)$ approaches~2 
as~$\eps\rightarrow 0$; see the red path in \autoref{fig:locally-fat-construction}(i).
This shortest path resembles the order-1 Hilbert curve, except that the middle link
has twice the length of the first and last link, whereas in the Hilbert curve all links have the same length.
\begin{figure}
\begin{center}
\includegraphics{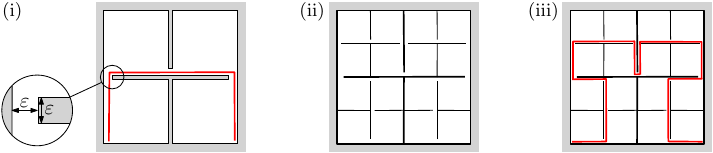}
\end{center}
\caption{(i) The basic building block~$P_1$ and the shortest path from $(0,0)$ to $(1,0)$ in~$P_1$.
(ii)~Schematic illustration of $P_2$. (iii) The shortest path from $(0,0)$ to $(1,0)$ in~$P_2$.}
\label{fig:locally-fat-construction}
\end{figure}

We will recursively construct $P_m$, for $m\geq 2$, such that the shortest 
path~$\pi((0,0),(1,0))$ from $(0,0)$ to $(1,0)$ resembles the order-$m$ Hilbert curve 
as $\eps\rightarrow 0$. 
We construct $P_m$ from $P_{m-1}$ 
as follows. We scale $P_{m-1}$ by $\frac12$ and  place four copies of it in a $2\times 2$ grid
whose cells have size $\frac{1}{2^{m-1}}$, where the south-west copy is rotated clockwise 
by $\frac\pi2$ and the south-east copy is rotated counter-clockwise by~$\frac\pi2$.
Finally, we create small openings in the top-left corner of the south-west copy,
in the bottom-right corner of the north-west copy, and in the 
top-right corner of the south-east copy. This ensures that $P_m$ is a simple polygon
and that the only way to go from $(0,0)$ to $(1,0)$ is to pass through these openings
and visit all four copies
See \autoref{fig:locally-fat-construction}(ii) for an illustration of~$P_2$. 
Note that the Euclidean diameter of $P_m$ equals $\sqrt 2$. Furthermore, 
the length $L_m$ of the shortest path~$\pi((0,0),(1,0))$ in $P_m$
(as $\eps\rightarrow 0$) satisfies
\[
L_m = 4 \cdot \left( \frac{1}{2}L_{m-1}\right) = 2 \cdot L_{m-1},
\]
where $L_1=2$. Thus, $L_m=2^m$. We obtain the following observation.
\begin{observation} \label{obs:Ln}
The polygon $P_m$ has $O(4^m)$ vertices and is such that $\diam(P_m)=\sqrt 2$ 
and $L_m\rightarrow 2^m$ as $\eps\rightarrow 0$, 
where $L_m$ is the length of the shortest path in $P_m$ from $(0,0)$ to~$(1,0)$.
\end{observation}
%
The next lemma shows that $P_m$ is locally fat.
\begin{lemma}\label{lem:locally-fat}
The polygon~$P_m$ is locally $\gamma$-fat for $\gamma=\frac{1}{8\pi}$.
\end{lemma}
\begin{proof}
    For simplicity we assume in the computations below that $\eps$ is infinitesimally small.
    Let $D(p,r)$ be a disk centered at a point $p\in P_m$ such that $D(p,r)$ 
    does not contain $P_m$ in its interior. Consider a hierarchical grid inside the 
    unit square~$[0,1]^2$, where the cells at level~$\ell$ have side length $\tfrac1{2^\ell}$ 
    and diameter $\tfrac{\sqrt{2}}{2^\ell}$. 
    Let $S_0,S_1,...,S_{m}$ be the squares of level $0,1,...,m$ containing $p$. 
    Since $P_m \not \subset D(p,r)$, we have $r< \sqrt{2}$. We consider three cases.
    \begin{itemize}
    \item Case~I: $r< \frac{\diam(S_m)}{2}$. 
        Then $D(p,r)$ does not contain~$S_m$. Moreover, since $S_m$ is a square at the deepest level,
        it doesn't contain any parts of $\partial P$ in its interior by construction. 
        This implies that $\area(D(p,r)\sqcap P_m)\geq \tfrac{\pi r^2}{4}$. 
        Hence,
        \[
        \area(D(p,r)\sqcap P_m)\geq \area(D(p,r)\sqcap S_m)\geq \frac{\pi r^2}{4} \geq \frac{1}{4}\cdot \area(D(p,r)).
        \]
    \item Case~II: $\frac{\diam(S_m)}{2} \leq r <\diam(S_m)$. 
         Then $D(p,r)$ contains more than $1/4$ of the area of~$S_m$, and so
        \[
        \area(D(p,r)\sqcap P_m) > \frac{1}{4} \cdot \left( \frac{1}{2^m}\right)^2  
            = \frac{1}{8\pi} \cdot \pi \left( \frac{\sqrt{2}}{2^m}\right)^2  
            > \frac{1}{8\pi} \cdot  \area(D(p,r)).
        \]
    \item Case~III: $r \geq \diam(S_m)$.
          Then there exists an $\ell$ with $0\leq \ell < n$ such that $\diam(S_{\ell+1}) \leq r \leq \diam(S_{\ell})$. 
          Since $p\in S_{\ell+1}$, we thus have $S_{\ell+1} \subset D(p,r)$. 
          Moreover, by construction, all points of $P_m$ within $S_{\ell+1}$ are connected via a path 
          that stays inside $S_{\ell+1}$. This means $ S_{\ell+1} \subset D(p,r)\sqcap P_m$.  
          Since $r \leq \diam(S_\ell)$, we also have $\area(D(p,r)) \leq  \pi \left( \tfrac{\sqrt{2}}{2^{\ell}} \right)^2$. Hence,
          \[
          \area(D(p,r)\sqcap P_m) \geq \left( \frac{1}{2^{\ell+1}} \right)^2 
          = \frac{1}{8\pi} \cdot \pi \left( \frac{\sqrt{2}}{2^{\ell}}\right)^2  
          \geq \frac1{8\pi} \cdot \area(D(p,r)),
          \]
          which concludes the proof.
        \end{itemize}     
\end{proof}
One may think that we now only need to show that $P_m$ does not have bounded
doubling dimension but, even though $\lim_{m\rightarrow \infty} \frac{\gdiam(P_m)}{\diam(P_m)} = \infty$, this is not the case. 
\begin{lemma}[restate=ddPn,name=]\label{lem:dd-Pn}
    The doubling dimension of the polygon $P_m$ is $O(1)$.
\end{lemma}
\begin{proof}
As in the proof of Lemma~\ref{lem:locally-fat}, we will assume that $\eps$ is infinitesimally small.
It will also be convenient to scale the polygon $P_m$ by a factor $2^m$. Thus $P_m$
consists of $2^m\times 2^m$ cells that are  unit squares. We number these cells $s_1,\ldots,s_k$, where $k={4^m}$,
in the order in which they are visited by the shortest path from the lower-left corner of $P_m$ 
to its lower-right corner. Observe that the only way to enter a cell~$s_i$ is from the
previous cell~$s_{i-1}$ or the next cell~$s_{i+1}$.

Now consider a geodesic disk~$D_g(p,r)$. 
   
We must show that $D_g(p,r)$ can be covered
by $O(1)$ geodesic disks of radius~$r/2$. Let $s_j$ be the cell containing the point~$p$. 
Then $D_g(p,r)$ consists of $z\geq 0$ consecutive cells $s_i,\ldots,s_{i+z-1}$  
that are fully contained in $D_g(p,r)$ plus at most four cells that are partially 
contained in $D_g(p,r)$; see \autoref{fig:linear-order-new} (i).  
\begin{figure}
\begin{center}
\includegraphics{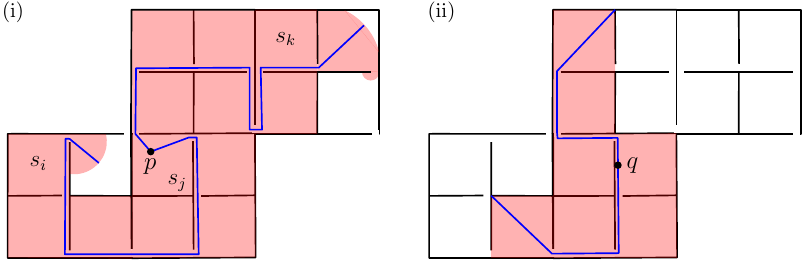}
\end{center}
\caption{Illustration for the proof of \autoref{lem:dd-Pn}. Note that only the relevant part of~$P_m$ is shown.}
\label{fig:linear-order-new}
\end{figure}
Because any Euclidean disk can be covered by seven Euclidean disks of half its radius and the
cells $s_t$ are convex, we know that $D_g(p,r) \cap s_t$
can be covered by at most seven geodesic disks of radius~$r/2$, for any cell $s_t$.
Hence, the part of $D_g(p,r)$ inside the four partially covered cells can be 
covered by 28~geodesic disks of radius~$r/2$.

Now consider the $z$ cells $s_{i},\ldots, s_{i+z-1}$ that are fully contained in~$D_g(p,r)$. 
Assume that~$z>0$, otherwise we are done.
Since $\gdiam(s_{i} \cup \cdots \cup s_{i+z-1}) \geq z$, we must have $r\geq \tfrac{z}{2}\geq \tfrac{1}{2}$.
We now have two cases.
\begin{itemize}
\item If $z\leq 7$ then we can cover each of the $z$ cells 
      by seven geodesic disks of radius~$r/2$, using at most 49 disks in total.
\item If $z\geq 8$ then we partition the sequence $s_{i},\ldots, s_{i+z-1}$ into eight 
      subsequences, each consisting of at most $\lceil z/8 \rceil$ cells.
      Note that any sequence of $z^*$ cells can be covered by a geodesic disk
      of radius $(z^*-2)/2+\sqrt{2}$. This is done by placing such a disk at the midpoint $q$ of a longest path between the first and last cell of that sequence; see Figure~\ref{fig:linear-order-new}(ii).
      Hence, a sequence of $\lceil z/8 \rceil$ cells
      can be covered by a geodesic disk of radius
      \[
      \frac{\lceil z/8 \rceil -2}{2} + \sqrt{2} 
        < \frac{z}{16} + (\sqrt{2}-\tfrac{1}{2}) 
        \leq  \frac{z}{16} +  \frac{z}{8}  < \frac{z}{4} \leq r/2
      \]
      Thus, in this case we use eight disks to cover the full sequence of~$z$ cells.
\end{itemize}
We conclude that $D_g(p,r)$ can be covered by at most 28+49 disks of radius~$r/2$. 
Hence, the doubling constant is at most~77,
and the doubling dimension is at most $\log 77 \approx 6.27$.
\end{proof}

The lemma above shows that by itself, the polygon $P_m$ does not show that there
are locally-fat polygons with unbounded doubling dimension. However, we can glue
copies of $P_m$ together to form a locally-fat polygon with unbounded doubling dimension.
\begin{figure}
\begin{center}
\includegraphics{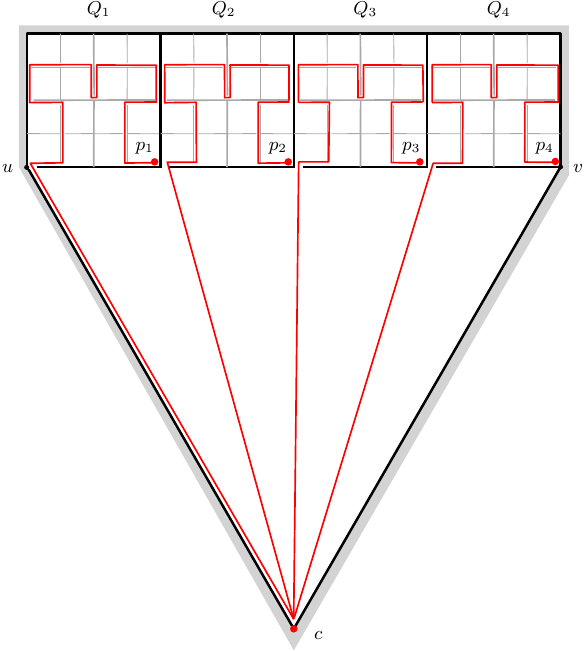}
\end{center}
\caption{A chain of scaled copies $Q_1,\dots,Q_{2^m}$ of the polygon $P_m$---in 
this example $m=2$---and the final polygon~$P^*_m$
obtained by attaching an equilateral triangle with apex~$c$.
The red curves show the paths from $c$ (slightly displaced for clarity) to the points~$p_i$.}
\label{fig:unbounded-dd}
\end{figure}
\begin{theorem} \label{thm:locallyfat-unbounded}
For any $n \geq 3$ there exists a simple polygon $P$ with $n$ vertices 
that is locally fat and whose doubling constant is $\Omega(n^{1/3})$. 
\end{theorem}
\begin{proof}
    Let $Q_1,Q_2,\dots,Q_k$ be $k :=2^m$ copies of $P_m$ scaled by a factor of $\frac{1}{2^m}$.
    By \autoref{obs:Ln} each $Q_i$ has geodesic diameter $1$.
    We place the $Q_i$ next to each other such that the right side of $Q_i$ coincides with 
    the left side of $Q_{i+1}$, for $1\leq i\leq k-1$; see \autoref{fig:unbounded-dd}.
    We make a small opening at the lower-left corner of each~$Q_i$. 
    Let $u$ be the lower-left corner of $Q_1$ and $v$ be the lower-right corner $Q_k$. 
    Clearly, $|uv| = 1$. Consider the point $c$ below $uv$ such that $\triangle cuv$ 
    is equilateral. The polygons $Q_1,\ldots,Q_k$, together with the segments $cu$ and $cv$
    define a polygon, which we denote by $P^*_m$. Since the $Q_i$ are locally fat, the
    polygon $P^*_m$ is fat as well. However, the geodesic disk $D(c,2)$ cannot be covered 
    by a constant number of geodesic disks of radius~$1$. Indeed, if $p_i$ is the point of 
    $Q_i$ at furthest geodesic distance from $c$, then any two $p_i,p_j$ have geodesic distance 
    $\gdist(p_i,p_j)>2$ and can thus never be covered by the same geodesic disk of radius one.
    Hence, the doubling constant of $P^*_m$ is at least $2^m$.

    Note that each $Q_i$ has $O(4^m)$ vertices. Hence, if $n$ denotes the number
    of vertices of $P^*_m$ then $n=O(8^m)$. Thus, the
    doubling constant of $P^*_m$ is $\Omega(n^{1/3})$.
\end{proof}

\subsection{$(\alpha,\beta)$-covered domains have bounded doubling dimension}
Recall that a domain~$P$, possibly with holes and a curved boundary, is $(\alpha,\beta)$-covered if for 
every $p\in \partial P$ there exists a triangle $T_p\subset P$ with
$p$ as a vertex such that $T_p$ is $\alpha$-fat---each angle is
at least~$\alpha$---and such that each side of $T_p$ has length at least $\beta\cdot \diam(P)$.
\begin{theorem} \label{thm:alpha-beta-dd} 
Any $(\alpha,\beta)$-covered domain~$P$ has bounded doubling dimension.
In particular, it is at most $\log c(\alpha,\beta)$, where
$c(\alpha,\beta) := \max \left( \left( \left\lceil\frac{48}{\sin\alpha} \right\rceil+1\right)^{2},  
                              \left(\left\lceil \frac{16}{\beta\sin\alpha}\right\rceil+1\right)^2\right)$.
\end{theorem}
\begin{proof}
Consider a geodesic disk $D := D_{\mathrm{g}}(p,r)$ with center~$p$ and 
radius~$r$ in~$P$. Let $B$ be the square of side $3r$ centered at $p$.  
We partition $B$ into a regular square grid $G$ with $g \times g$ cells, where $g := \max ( \lceil\frac{48}{\sin\alpha}\rceil, \lceil \frac{16}{\beta\sin\alpha}\rceil)$; see \autoref{fig:bounded-dd}. Let $s$ denote the side length of a cell. Then we can see that $s = \frac{3r}{g} \leq \frac{m\cdot \sin\alpha}{16}$, where $ m := \min(r, \beta\cdot\diam(P))$.
\begin{figure}
\begin{center}
\includegraphics{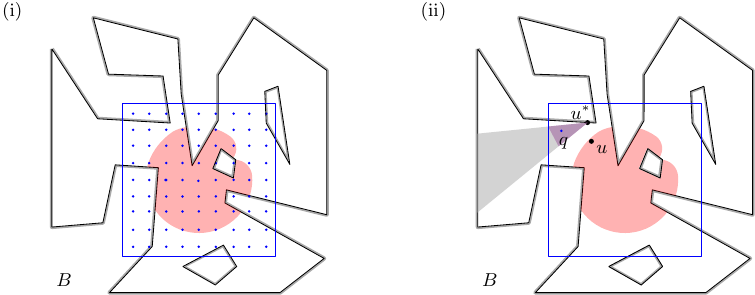}
\end{center}
\caption{Construction of the grid~$G$. The geodesic disk $D$ is shown in red. The figure is schematic and not to scale.}
\label{fig:bounded-dd}
\end{figure}

Now consider the set $\mathcal{D} = \{D_{\mathrm{g}}(q,r/2) : q \in G \cap P\}$ of geodesic balls of radius~$r/2$ centered at the grid points inside~$P$. Note that $|\mathcal{D}|\leq c(\alpha,\beta)$.
To prove \autoref{thm:alpha-beta-dd}, it suffices to show that the disks in $\mathcal{D}$ cover $D$.
%
%
In other words, we need to show that for any point $u \in D$
there exists a grid point $q \in G \cap P$ with $\gdist(u,q) \le r/2$. 
We need the following well-known fact. For completenes, we give a proof.
\begin{fact}[restate=triangleingrid,name=]\label{obs:grid-triangle}
    Let $G$ be a regular grid in $\mathbb{R}^2$ whose cells have side length~$x$. 
Let $T$ be an isosceles triangle whose two equal sides have length~$t$ and 
meet at angle~$\phi$. If $t \ge \frac{4x}{\sin\phi}$, then $T$ contains a grid point of~$G$.
\end{fact}

\begin{proof}
    Every point of $\mathbb{R}^2$ is within Euclidean distance $x$---in
    fact, within distance~$\frac{\sqrt{2}}{2}x$---of some grid point.
    It therefore suffices to show that $r_{\mathrm{in}}$, the radius of the incircle of 
    $T$, is at least $x$. 
    We have $r_{\mathrm{in}} = \frac{2 \cdot \mathrm{area}(T)}{\mathrm{per}(T)}$.
    Using $\mathrm{area}(T) = \frac{1}{2} t^{2} \sin\phi$ 
    and $\mathrm{per}(T)= 2t + 2t\sin(\phi/2)$ we obtain
    \[
        r_{\mathrm{in}}
            = \frac{t^2 \sin\phi}{2t(1+\sin(\phi/2))}
            \ge \frac{t\sin\phi}{4}.
    \]
    Thus, if $t \ge 4x/\sin\phi$ then $r_{\mathrm{in}} \ge x$, which finishes the proof.
\end{proof}
We now consider two cases.
\begin{itemize}
\item \emph{Case 1: $D(u,r/4) \subset P$.}
    Since $s\leq   r\cdot (\sin\alpha/16) \leq r/4$, 
    the disk $D(u,r/4)$ contains a grid point $q$. 
    Since $D(u,r/4) \subset P$ we therefore have $\gdist(u,q) =|uq| \le r/4 \le r/2$.
\item \emph{Case 2: $D(u,r/4)$ intersects $\partial P$.}
    This case is illustrated in Figure~\ref{fig:bounded-dd}(ii).
    Let $u^*\in\partial P$ be a point closest to~$u$. Then $|u u^{*}| < r/4$
    and $uu^*\subset P$. Let $T_{u^{*}}\subset P$ be an $\alpha$-fat triangle 
    with vertex $u^*$ and all of whose edges have length at least $\beta\cdot\diam(P)$, 
    which is guaranteed to exist because $P$ is $(\alpha,\beta)$-covered. 

    \smallskip
    If $r < \beta \cdot \diam(P)$ then every side of $T_{u^{*}}$ has length at least $r$.
    Thus, $T_{u^{*}}$ contains an isosceles $\alpha$-fat triangle $T$ with 
    apex~$u^*$ and side length $t=r/4$. Because $T\subset P\cap B$ and
    $ r/4 = m/4 \geq 4s/\sin\alpha$,
    the triangle~$T$ contains a grid point $q$ by \autoref{obs:grid-triangle}. 
    Since $|u^{*} q| \le r/4$, we thus have
    \[
        \gdist(u,q) \le |uu^{*}| + |u^{*}q| \le \frac{r}{4} + \frac{r}{4} = \frac{r}{2}.
    \]
    On the other hand, if $r \ge \beta \cdot \diam(P)$ then $T_{u^{*}}$ contains 
    an isosceles $\alpha$-fat triangle with side 
    length $t' = \beta \cdot\diam(P)/4 = m/4 \geq4s/\sin\alpha$. Again by
    \autoref{obs:grid-triangle} it contains a grid point $q'$, and since $t' \le r/4$ we also have $\gdist(u,q') \le r/2$.
\end{itemize}
We conclude that every $u\in D$ lies in some disk of $\mathcal{D}$, which
finishes the proof of \autoref{thm:alpha-beta-dd}.
\end{proof}

\subparagraph{Applications.}
The fact that $(\alpha,\beta)$-covered domains have bounded doubling
dimension immediately gives a plethora of results that improve on the state-of-the-art
for arbitrary (non-fat) domains. We mention results on spanners and WSPDs, because they form the basis of many other results.
\begin{corollary} \label{cor:applications}
Let $P$ be an $(\alpha,\beta)$-covered domain and let $S\subset P$ be a set of $m$ points.
\begin{enumerate}[(i)]
\item Let $G=(S,E)$ be the complete graph on~$S$ such that the weight of an edge~$(p,q)\in S\times S$ 
      is its geodesic distance $\gdist(p,q)$. Then there exists a $(1+\eps)$-spanner of $G$ 
      consisting of $m (1/\eps)^{O(1)}$ edges.
\item For any fixed $s>1$, there exists an $s$-well seperated pair decomposition (WSPD) of $S$ of size $m\cdot s^{O(1)}$.
\end{enumerate}
\end{corollary}
\begin{proof}
Part~(i) follows from the result of Gao, Guibas, and Nguyen~\cite{DBLP:journals/comgeo/GaoGN06}, who showed that a set of $m$ points in a space of doubling dimension~$\mathit{dim}$ has a $(1+\eps)$-spanner of size $m (1/\eps)^{O(\mathit{dim})}$ edges. For part~(ii), we use the WSPD construction by Har-peled and Mendel~\cite{wspd}, which has size $m\cdot s^{O(\mathit{dim})}$.
 \end{proof}
Note that the bound on the spanner size in part (i) of \autoref{cor:applications} is linear
in $m$ and does not depend on the size of the polygon. A similar result is not possible 
for non-fat polygons: for any $\eps>0$ and any $m$ there is a polygon $P$ and a set 
of $m$ points such that any spanner of subquadratic size has spanning ratio at least~$2-\eps$.
(Take the polygon of \autoref{fig:fatness}(i) and put a point in each of the spikes.)
Also note that there are more results on spanners in spaces of bounded doubling dimension---for
example on the lightness of the spanner~\cite{DBLP:conf/soda/BorradaileLW19}---and that any such result immediately 
applies to point sets in an $(\alpha,\beta)$-covered polygonal domain.

\section{The perimeter of geodesically convex regions in fat polygonal domains}
\label{sec:geoconvex}

Recall that Bose et al.~\cite{bose11} showed that for any $(\alpha,\beta)$-covered simple polygon $P$ there exists a constant $\mu(\alpha,\beta)$, depending only on $\alpha$ and $\beta$, such that
$\per(P)\leq \mu(\alpha,\beta)\cdot\diam(P)$.
\autoref{thm:bounded-ratio} generalizes this result to arbitrary geodesically convex polygonal sets in~$P$, and extends it from simple polygons to general polygonal domains. 
As we will see in the next section, this is useful in several applications.
Note that we cannot hope to generalize the result to locally-fat polygons.
Indeed, there are locally-fat polygons for which we do not even have $\per(P)=O(\diam(P))$---the
polygon~$P_n$ from \autoref{sec:dd} is an example.
\begin{theorem}\label{thm:bounded-ratio}
Let $P$ be an $(\alpha,\beta)$-covered polygonal domain, possibly with holes.
Let $R \subseteq P$ be a geodesically convex set in $P$ such that
$\partial R \setminus \partial P$ is polygonal.
Then there exists a constant $\nu(\alpha,\beta)$ such that
$\per(R) \le \nu(\alpha,\beta)\cdot \diam(R)$.
\end{theorem}
To prove \autoref{thm:bounded-ratio}, define $\Gamma_1 := \partial R \cap \partial P$ and 
$\Gamma_2 := \partial R \setminus \Gamma_1$, and note that $\per(R)=\|\Gamma_1\|+\|\Gamma_2\|$. Refer to Figure~\ref{fig:example-geodesically-convex} for an illustration of these sets. 
We will bound $\|\Gamma_1\|$ and $\|\Gamma_2\|$ separately.

\begin{figure}
\begin{center}
\includegraphics{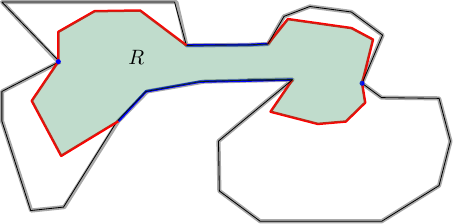}
\end{center}
\caption{The shaded region $R$ is geodesically convex. The blue curves (and points) are the connected components of $\Gamma_1$, while the red curves are the connected components of $\Gamma_2$.}
\label{fig:example-geodesically-convex}
\end{figure}

\subparagraph{Bounding $\|\Gamma_1\|$.}
Let $S$ be an axis-aligned square of side length~$\diam(R)$ that contains~$R$.
We partition $S$ into a regular $g\times g$ grid~$\mathcal{G}$, where
$g:=\left\lceil \tfrac{2}{\beta\sin\alpha}\right\rceil$.
Let $s:=\diam(R)/g$ be the side length of a grid cell.
Then $s\le \tfrac{\beta\sin\alpha}{2}\cdot\diam(R)$, and each cell~$C\in\mathcal{G}$ satisfies
\begin{equation}\label{eq:cell-diameter}
    \diam(C) = \sqrt{2}\,s < \beta\sin\alpha\cdot \diam(R).
\end{equation}
Moreover,
\begin{equation}\label{eq:grid-perimeter}
\sum_{C\in\mathcal{G}} \per(C)
 \;=\; g^2\cdot 4s
 \;=\; 4g \cdot \diam(R)
 \;=\; O\!\left(\tfrac{1}{\alpha\beta}\cdot \diam(R)\right).
\end{equation}
We will need the following lemma. 
\begin{lemma}\label{lem:triangle-ray}
    Let $p\in \partial P \cap C$ for some cell $C\in \mathcal{G}$. Then any 
    ray $\rho$ emanating from $p$ and going into the witness triangle $T_p$, 
    will hit $\partial C$ before exiting $T_p$.
\end{lemma}
\begin{proof}
It suffices to prove that the side $e$ of $T_p$ opposite to $p$ lies fully outside $C$. 
Assume for a contradiction that there exists a point $q\in e$ that is contained in~$C$. 
Since $p\in C$ and $q\in C$, we have that $\diam(C)\geq |pq|$. Moreover, $|pq| \geq h_p$,
where $h_p$ is the height of $T_p$ from~$p$. Since $h_p \geq \sin\alpha\cdot\beta \cdot \diam(R)$, we have 
\[
\diam(C) \geq |pq|\geq h_p \geq \sin\alpha\cdot\beta \cdot \diam(R),
\]
which contradicts Inequality~(1).
\end{proof}
Let $\mathcal{D} = \{d_i := i \frac{\alpha}{4} \mid i\in \mathbb{Z} \mbox{ and } 0\leq i\leq \frac{8\pi}{\alpha}\}$ 
be a set of \emph{canonical directions}.\footnote{With a slight abuse of terminology, we identify an angle $d_i\in \D$ with
the direction of a vector whose counterclockwise angle with the positive $x$-axis is~$d_i$.}
For a point $p\in P$ and a direction $d$, let $\rho(p,d)$ be the ray starting from~$p$ and going into
the direction~$d$. Let $C$ be a cell in $\mathcal{G}$ and consider a point $p\in \partial P\cap C$ 
that lies in the relative interior of an edge $e$ of~$\partial P$. 
We say a direction $d_i\in\mathcal{D}$ is \emph{good} for~$p$ 
if the following holds for the ray~$\rho(p,d_i)$:
\begin{enumerate}[(i)]
\item $\rho(p,d_i)$ makes an angle of at least $\alpha/4$ with the edge~$e$, and 
\item $\rho(p,d_i)$ hits $\partial C$ before hitting $\partial P$.
\end{enumerate}
The following lemma will allow us to bound $\|\Gamma_1\|$.
\begin{lemma}\label{lem:cell-projection}
Let $C$ be a cell in $\mathcal{G}$ and let $X\subset \partial P\cap C$ be a 
point set that consists of finitely many connected components. Assume that every $p\in X$ 
that is not a vertex of $P$, has at least one good direction in $\mathcal{D}$.
Then
\[
    \|X\|\leq c(\alpha)\cdot \per(C)
    \qquad\text{where}\qquad
    c(\alpha):=\tfrac{8\pi}{\alpha\sin(\alpha/4)}.
\]
\end{lemma}
\begin{proof}

For each $d_i\in\mathcal{D}$, let $X_i\subseteq X$ be the set of points where $d_i$ is good.
Note that $ \|X\| \leq \sum_{d_i\in\mathcal{D}} \|X_i\|$.

Fix a direction $d_i$ and consider $X_i$.
Assume wlog that $d_i$ is the vertically upward direction.
We define a mapping $f_i:X_i\rightarrow \partial C$ such that $f_i(p)$ is the point 
where $\rho(p,d_i)$ hits~$\partial C$; see Figure~\ref{fig:mapping}.

\begin{figure}
\begin{center}
\includegraphics{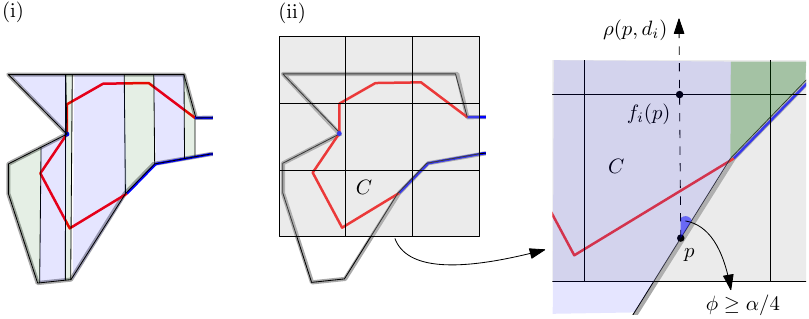}
\end{center}
\caption{(i) The vertical decomposition of part of the polygon from Figure~\ref{fig:example-geodesically-convex}. (ii) Part of the grid $\mathcal{G}$. In the cell $C$, the direction $d_i$ is good for $p$ (and in this example it is good for all points of $\partial P\cap C)$.}
\label{fig:mapping}
\end{figure}

Observe that $f_i$ is injective. Indeed, if $q = f_i(p)$ then $p$ is the 
unique point in $X_i$ that is hit by a vertically downward ray from $q$.
Let $\vd(P)$ be the vertical decomposition of~$P$, and consider a trapezoid~$\Delta\in\vd(P)$. 
Let $X_i(\Delta)$ be the part of $X_i$ contained in the bottom side of $\Delta$---the
top side cannot contain any part of $X_i$---and assume that $X_i(\Delta)\neq\emptyset$. 
Observe that $f_i(p)\in\Delta\cap \partial C$ for all $p\in X_i(\Delta)$.
Since the angle between $\rho(p,d_i)$ and the bottom side of~$\Delta$ is at least~$\alpha/4$,
we thus have $\|X_i(\Delta)\| \leq \tfrac{1}{\sin(\alpha/4)} \cdot \|\partial C\cap \Delta\|$.
Hence,
\[
\|X_i\| = \sum_{\Delta} \|X_i(\Delta)\| 
      \leq \tfrac{1}{\sin(\alpha/4)} \cdot \|\partial C \cap \Delta\| 
      = \tfrac{1}{\sin(\alpha/4)} \cdot \|\partial C \|.
\]
Finally, we bound $\|X\|$ by summing over all canonical directions:
\[
    \|X\| \leq \sum_{d_i\in\mathcal{D}} \|X_i\|
    \leq |\mathcal{D}| \cdot \tfrac{1}{\sin(\alpha/4)} \cdot \|\partial C\|
    \leq \tfrac{8\pi}{\alpha\sin(\alpha/4)}\cdot \|\partial C\|.
\]
This proves the lemma.
\end{proof}
We can now bound~$\|\bd P\cap S\|$.
Note that this immediately implies the same bound for $\|\Gamma_1\|$, since $\Gamma_1\subseteq \partial P\cap S$.
\begin{lemma}\label{lem:boundary-in-square}
$\|\partial P \cap S\| = O\!\left(\tfrac{1}{\beta\alpha^3}\cdot \diam(R)\right)$.
\end{lemma}
\begin{proof}
Fix a grid cell $C\in \mathcal{G}$ and consider $X:=\partial P\cap C$.
Let $p\in X$ be a point that lies in the relative interior of an edge $e$ of $\partial P$, 
and let $T_p$ be its witness triangle. Since the angle of $T_p$ at $p$ is at least $\alpha$, there are 
at least three canonical directions $d_{i-1},d_i,d_{i+1}$ such that the rays 
$\rho(p,d_{i-1})$, $\rho(p,d_i)$, and $\rho(p,d_{i+1})$ go into~$T_p$.
The ray $\rho(p,d_i)$ must therefore make an angle of at least~$\alpha/4$ with the
edge~$e$. Lemma~\ref{lem:triangle-ray} implies that the ray $\rho(p,d_i)$ will hit $\partial C$ 
before exiting $T_p$. As a result, the direction $d_i$ is good for $p$. 
Hence, we can apply \autoref{lem:cell-projection} to $X$ and conclude that
\[
\|\partial P\cap C\| = \|X\| \le \tfrac{8\pi}{\alpha\sin(\alpha/4)}\cdot \per(C).
\]
Summing over all cells $C\in \mathcal{G}$ and using Equation~(\ref{eq:grid-perimeter}) gives
\[
\|\partial P\cap S\|
\le \tfrac{8\pi}{\alpha\sin(\alpha/4)}\cdot \sum_{C\in G}\per(C)
= O\!\left(\tfrac{1}{\beta\alpha^3}\cdot \diam(R)\right).
\]
\end{proof}
As mentioned, the fact that $\Gamma_1\subset \partial P \cap S$ immediately gives us the following corollary.
\begin{corollary}\label{lem:length-gamma1}
    The total length of $\Gamma_1$ is $O(\frac{1}{\beta \alpha^3}\cdot \diam(R))$.
\end{corollary}

\subparagraph{Bounding $\|\Gamma_2\|$.}
We now bound the contribution to the perimeter coming from $\Gamma_2$,
the part of $\partial R$ that lies in the interior of~$P$.

For a grid cell $C\in\mathcal{G}$, let $E_{\mathrm{v}}(C)$ be the set of edges of $\Gamma_2\cap C$
that are \emph{mostly vertical}, in the sense that their angle with the $x$-axis is at least~$\pi/4$.
Define $E_{\mathrm{h}}(C)$ analogously for \emph{mostly horizontal} edges. Note that
any edge is mostly vertical or mostly horizontal (or both, when it's slope is exactly~1).
\begin{lemma}\label{lem:gamma2-in-cell}
For any cell $C\in\mathcal{G}$ we have
\[
\|\Gamma_2\cap C\| = O\!\bigl(\per(C) + \|\partial P \cap C\|\bigr).
\]
\end{lemma}
\begin{proof}
We prove the bound for $E_{\mathrm{v}}(C)$; the bound for $E_{\mathrm{h}}(C)$ is symmetric.
Since $\|\Gamma_2\cap C\| = \|E_{\mathrm{v}}(C)\| + \|E_{\mathrm{h}}(C)\|$, the lemma then follows.

Consider the edges of $E_{\mathrm{v}}(C)$ that bound~$R$ from the left, that is, the interior of $R$ lies locally to their right. 
Let $p$ be a point on such an edge and shoot a ray $\rho$ from~$p$ horizontally to the right.
We claim that $\rho$ cannot hit another edge of $E_{\mathrm{v}}(C)$ that bounds $R$ from the left
before hitting $\partial P$, and thus it either hits $\partial P$ first or exits~$C$.

Indeed, if $\rho$ would hit another such edge at a point~$q$ before hitting $\partial P$,
then the horizontal segment $pq$ is contained in~$P$. Hence, $pq$ is a shortest path in~$P$
between $p$ and~$q$, and because $R$ is geodesically convex, we must have $pq\subseteq R$.
This contradicts that $q$ lies on an edge that bounds~$R$ from the left.

\begin{figure}
\begin{center}
\includegraphics{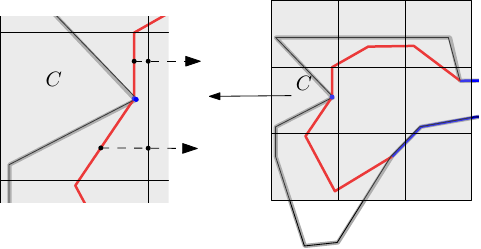}
\end{center}
\caption{The two red edges in cell $C$ bound $R$ from the left and are mostly vertical. The two shown horizontal rays hit $\partial C$ before hitting $\partial P$.}
\label{fig:mapping-gammatwo}
\end{figure}

We now charge any such point $p$ on a left-bounding edge to where the ray $\rho$ hits $\partial P$
(if it hits $\partial P$ first) or to where it exits~$C$
(if it does not hit $\partial P$ inside~$C$); see Figure~\ref{fig:mapping-gammatwo}. This is analogous to the method used in Lemma~\ref{lem:cell-projection}: since we consider mostly vertical edges, we know that $\rho$ forms an angle of at least $\pi/4$ with the corresponding edge of $\partial P$ (whereas in Lemma \ref{lem:cell-projection} this bound was $\alpha/4$). Therefore, by using a horizontal decomposition of $P$, we can similarly argue that the total length charged to $\partial P$
is $O(\|\partial P \cap C\|)$ and the total length charged to $\partial C$ is $O(\per(C))$.
In total we have
\[
\|E_{\mathrm{v}}(C)\| = O\!\bigl(\per(C) + \|\partial P \cap C\|\bigr).
\]
A similar argument applies to the edges of $E_{\mathrm{v}}(C)$ that bound $R$ from the right.
\end{proof}
Summing the bound from \autoref{lem:gamma2-in-cell} over all grid cells yields
\[
\|\Gamma_2\|
= \sum_{C\in\mathcal{G}} \|\Gamma_2\cap C\|
= O\!\left(\sum_{C\in\mathcal{G}} \per(C) + \|\partial P \cap S\|\right)
= O\!\left(\frac{1}{\beta\alpha^3}\cdot \diam(R)\right),
\]
where the last step uses Lemma~\ref{lem:boundary-in-square} and the bound on $\sum_{C\in\mathcal{G}}\per(C)$
given in Equation~(\ref{eq:grid-perimeter}).
\medskip

Since $\per(R)=\|\Gamma_1\|+\|\Gamma_2\|$,
we conclude that
\[
\per(R) = O\!\left(\tfrac{1}{\beta\alpha^3}\cdot \diam(R)\right).
\]
This proves \autoref{thm:bounded-ratio} with $\nu(\alpha,\beta)=O\!\left(\frac{1}{\beta\alpha^3}\right)$.

\subparagraph{Application to coresets for furthest-neighbor queries.}
Let $S$ be a point set in a simple polygon~$P$.
Recall that an $\eps$-coreset of $S$ for furthest-neighbor queries 
is a set $C\subseteq S$ such that for any query point $q \in P$ we have 
$\gdist(q,\fn(q,C)) \geq (1-\varepsilon) \cdot \gdist(q,\fn (q,S))$,
where $\fn(p,Q)$ denotes the furthest neighbor of $q$ in a set $Q$. 
By applying Theorem~\ref{thm:bounded-ratio} for $R= \rch(S)$, we obtain the following result.
\begin{corollary}[restate=coreset,name=]\label{cor:coreset}
Let $S$ be a point set of size $m$ in an $(\alpha,\beta)$-covered simple polygon~$P$ with $n$ vertices.
For any $0<\varepsilon \leq 1$, there exists an $\eps$-coreset $C\subseteq S$ of size $O(1/\varepsilon)$
for furthest-neighbor queries. The coreset $C$ can be constructed in $O(n+m\log(n+m))$ time.
\end{corollary}
\begin{proof}
For any query point~$q$, its furthest neighbor $\fn(q,S)$ is a vertex of the relative convex 
hull~$\rch(S)$ of $S$ inside~$P$~\cite{DBLP:conf/compgeom/BergT24}.
Moreover, $\gdist(q,\fn(q,S))\geq \tfrac{1}{2} \gdiam(\rch(S))$. Hence, we can compute an $\eps$-coreset as follows.
First, we compute $\rch(S)$ in $O(n+m\log(n+m))$ time~\cite{toussaint-rch,DBLP:journals/dcg/Chazelle91,DBLP:journals/jcss/GuibasH89}. 
Then we traverse $\partial\,\rch(S)$ to select
a subset $C$ of $O(1/\eps)$ vertices of $\partial\,\rch(S)$ such that for any vertex $p'$ of $\rch(S)$, 
there is a vertex~$p\in C$ whose distance to $p'$ along $\partial\,\rch(S)$
is at most~$\tfrac{\eps}{2 \nu(\alpha,\beta)}\cdot \per(\rch(S))$, where $\nu(\alpha,\beta)$ is the constant
in \autoref{thm:bounded-ratio}.

To see that $C$ is an $\eps$-coreset, recall that $\per(\rch(S)) \leq \nu(\alpha,\beta) \cdot \diam(\rch(S))$
by \autoref{thm:bounded-ratio}. Now consider a query point $q$. Let $p$ be a point in~$C$ whose
distance to $\fn(q,S)$ along $\partial\,\rch(S)$ is at most~$\tfrac{\eps}{2 \nu(\alpha,\beta)}\cdot \per(\rch(S))$. 
Then $\gdist(p,\fn(q,S))\leq \tfrac{\eps}{2 \nu(\alpha,\beta)}\cdot \per(\rch(S))$, and so
\begin{align*}
\gdist(q,\fn(q,S)) & \leq \  \gdist(q,p) + \gdist(p,\fn(q,S)) \\
              & \leq \  \gdist(q,\fn(q,C)) + \tfrac{\eps}{2 \nu(\alpha,\beta)}\cdot \per(\rch(S)) \\
              & \leq \   \gdist(q,\fn(q,C)) + \eps\cdot \tfrac{1}{2} \cdot \diam(\rch(S)) \\
              & \leq \   \gdist(q,\fn(q,C)) + \eps\cdot \gdist(q,\fn(q,S)),
\end{align*}
which proves that $C$ is an $\eps$-coreset.
\end{proof}

\section{Closest pair in an $(\alpha,\beta)$-covered polygon}
\label{sec:closest-pair}
Let $Q$ be a set of $m$ points contained in an $(\alpha,\beta)$-covered polygon~$P$. The closest-pair problem asks for
a pair of points $p,q\in Q$ minimizing the geodesic distance $\gdist(p,q)$. Using Theorem~\ref{thm:bounded-ratio}
we prove the following lemma, which will allow us to adapt the classical linear-time closest-pair algorithm for 
the plane~\cite{dietzfelbinger97} to the geodesic metric in $P$.
\begin{figure}
\begin{center}
\includegraphics{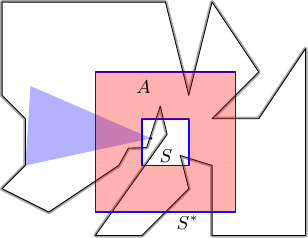}
\end{center}
\caption{Illustration for the proof of \autoref{lem:square-partition}.}
\label{fig:partition}
\end{figure}
\begin{lemma}\label{lem:square-partition}
Let $P$ be an $(\alpha,\beta)$-covered simple polygon, and let $S$ be a
square that intersects the interior of $P$. Let $C_1,\dots,C_m$ be the connected components of
$S\cap P$. Then there exists a partition of $\{C_1,\dots,C_m\}$ into $t= O(1)$
classes $\mathcal{C}_1,\dots,\mathcal{C}_t$ such that the following holds:
\begin{quotation}
\noindent There is constant $M$ such that for every class $\mathcal{C}_j$ and every pair of components
$C_a,C_b\in\mathcal{C}_j$ we have
$\gdist(x,y)\leq M \cdot \diam(S)$ for all points $x\in C_a$ and $y\in C_b$.
\end{quotation}
\end{lemma}
\begin{proof}
If $S\subset P$ then $m=1$ and the statement is trivial since any two
points $x,y\in S\cap P$ can be connected by a segment of length at most~$\diam(S)$. 
Hence, we assume that $\partial P$ intersects the interior of $S$. 
Let $\ell$ be the side-length of~$S$. 
We have two cases.

\medskip\noindent
\emph{Case 1: $\ell \ge \frac{\beta\cdot \diam(P)}{8}$.} Then, for any $x,y\in S\cap P$,
\[
  \gdist(x,y)\le \gdiam(P) 
          \le \per(P) \leq \mu(\alpha,\beta) \cdot\diam(P)
          \le \tfrac{8\mu(\alpha,\beta)}{\beta}\cdot \ell
          \le \tfrac{8\mu(\alpha,\beta)}{\beta}\cdot \diam(S).
\]
So in this case we may put all components $C_i$ into a single class and the lemma holds.

\medskip\noindent
\emph{Case 2: $\ell < \frac{\beta\cdot\diam(P)}{8}$.}

Observe that any connected component $C_i$ of $S\cap P$, is a geodesically convex region. Indeed, if there were $x,y\in C_i$ such that a shortest path $\pi(x,y)$ leaves $C_i$, 
then $\pi(x,y)$ must exit  $S$ at a point $a$ and re-enter $S$ at a point $b$. 
Let $\pi(a,b) \subset \pi(x,y)$ be the subpath of $\pi(x,y)$ from~$a$ to~$b$.
Since $C_i$ is connected, there is a path $\pi'(a,b) \subset C_i$. Since $P$ is
a simple polygon, the cycle $\pi'(a,b)\cup \pi(a,b)$ does not contain a part of $\partial P$.
But then we can replace $\pi(a,b)$ by the part of $\partial S$ that lies
inside the cycle, and obtain a shorter path from $x$ to $y$, which is a contradiction.
Hence, $C_i$ must be a geodesically convex, as claimed.
Applying Theorem~\ref{thm:bounded-ratio} with $R:=C_i$ we have $\per(C_i) \le \nu(\alpha,\beta)\cdot\diam(C_i)$. 
Since $C_i\subseteq S$, we have $\diam(C_i)\le \diam(S)$, and hence
for any $x,y\in C_i$ 
\begin{equation}\label{eq:path-inside}
  \gdist(x,y)
  \le \per(C_i)
  \le \nu(\alpha,\beta)\cdot\diam(C_i)
  \le \nu(\alpha,\beta)\cdot\diam(S).
\end{equation}
Since $\partial P$ meets the interior of $S$,
every component $C_i$ satisfies $C_i\cap\partial P\neq\emptyset$. For each
$i$ choose a point $p_i\in C_i\cap\partial P$ and let $T_i$ be its witness
triangle. Let $S^*$ be the square with the same center as $S$ and with side length $3\ell$, 
and let $A := S^* \setminus \myint(S)$ be the annulus between the two squares.
The maximum Euclidean distance from a point in $S$ to a point in $A$ is $2\cdot \diam(S)$; see Figure~\ref{fig:partition}. 
An edge of $T_i$ incident to $p_i$ has length at least $\beta\cdot\diam(P) \geq 8\ell \geq 4\cdot\diam(S)$. 
This implies that the two sides of $T_i$ incident to $p_i$ exit $S^*$, and because the angle between them is 
at least $\alpha$ we can conclude that $\area(T_i\cap A) = \Omega(\area(A))$. This implies that we can stab 
the collection of objects $\{T_i\cap A\}_{i=1}^m$, with a constant number of points $q_1,q_2,...,q_t$.
Our classes $\mathcal{C}_1,\mathcal{C}_2,...,\mathcal{C}_{t}$ are now defined by assigning
each component $C_j$ to the class~$\C_i$ such that $i$ is the smallest index for which $q_i$ stabs~$T_j\cap A$.

Now fix a class $\mathcal{C}_j$ and components $C_a,C_b\in\mathcal{C}_j$.
Let $x\in C_a$ and $y\in C_b$. To bound $\gdist(x,y)$, consider the path
$\pi = \pi(x,p_a)\cup p_a q_j\cup q_jp_b \cup \pi(p_b,y)$.
By Inequality (\ref{eq:path-inside}) we have 
\[
\|\pi(x,p_a) \| + \|\pi(p_b,y)\|=\gdist(x,p_a)+\gdist(y,p_b) = 2\nu(\alpha,\beta) \cdot \diam(S).
\]
Moreover the points $p_a,p_b,q_j$  lie in $S^*$, which implies that 
$|p_a q_j| \le \diam(S^*) =3\cdot\diam(S)$ and $ |q_jp_b| \le 3\cdot\diam(S)$.
Therefore $\gdist(x,y) \le (2\nu(\alpha,\beta)+6) \cdot \diam(S)$.
\medskip

Combining the two cases, the lemma holds for all squares $S$.
\end{proof}
\subparagraph{The algorithm.}
We adapt the algorithm by Dietzfelbinger~\etal~\cite{dietzfelbinger97}
that solves the problem in expected linear time in the Euclidean plane.
We process the points of $Q$ in random order, and maintain the current minimum distance
$\delta$ and a grid of squares of side length $s:=\delta/(2M)$, stored in a hash table.
As long as $\delta$ is the smallest geodesic distance among the processed
points, \autoref{lem:square-partition} implies that each grid cell contains at most $t=O(1)$ points.
Indeed, if $|X(S)|>t$ then two points of $X(S)$ would lie in
components belonging to the same class of Lemma~\ref{lem:square-partition},
implying $\gdist(x,y)\le M\cdot\diam(S)=\tfrac{\sqrt2}{2}\delta<\delta$, a
contradiction.

Now let $q_1,\ldots,q_m$ be the points from $Q$, in random order, and
suppose that the insertion of some point~$q_j$ changes the closest pair. 
Let $S$ be the cell of the current grid that contains the new point~$q_j$,
and let $q_i,q_j$ be the new closest pair. Since $\gdist(q_i,q_j)< \delta$ we know that 
$q_i$ must lie in cell $S'$ such that the Euclidean distance between $S$ and $S'$ is
smaller than~$\delta$. Thus, to find $q_i$ we can restrict our
attention to the $O(M^2)$ cells within distance $\delta$ of $S$ (including $S$ itself).
Since each cell contains at most $t$ points, only $O(1)$ candidates must be checked,
which we can do in $O(\log n)$ time using a shortest-path query.
When a smaller distance is found, we update $\delta$ and rebuild the grid.
The analysis of the expected running time is identical to the Euclidean case
\cite{dietzfelbinger97}, except that computing the distance between two given points incurs a logarithmic
overhead (after preprocessing $P$ for shortest-path queries in $O(n)$ time~\cite{DBLP:journals/jcss/GuibasH89}).
We obtain the following theorem.
\begin{theorem}
Let $P$ be an $(\alpha,\beta)$-covered simple polygon with $n$ vertices and
let $Q$ be a set of $m$ points in $P$. A closest pair of $Q$ under the geodesic distance
can be found in expected time $O(n+\kappa(\alpha,\beta)m\log n)$, where
$\kappa(\alpha,\beta)$ is a constant depending only on $\alpha$ and $\beta$.
\end{theorem}

\bibliography{ref}





\end{document}